\documentclass[12pt]{article}
\usepackage{amsmath, amsthm, amscd, amsfonts, amssymb, graphicx, color}
\usepackage[bookmarksnumbered, colorlinks]{hyperref}
\usepackage{float}
\usepackage{lipsum}
\usepackage{afterpage}
\usepackage[labelfont=bf]{caption}
\usepackage[nottoc,notlof,notlot]{tocbibind}
\usepackage{fancyhdr}
\usepackage{array,tabularx,booktabs}
\newcolumntype{Y}{>{\raggedright\arraybackslash}X}

\newcommand{\keyword}[1]{\par\noindent \textbf{Keywords:} #1 }

\newtheorem{defi}{Definition}[section]

\newtheorem{theorem}[defi]{Theorem}

\title{Control-Translated Finsler-type structure and Anisotropic Ginzbor-Landau models}

\author
{Y. Alipour Fakhri\thanks{Corresponding Author:
Faculty of Basic Sciences,
Department of Mathematics, Payame Noor University, Tehran, Iran.  E-mail: y\_ alipour@pnu.ac.ir}
}

\begin{document}
\maketitle

\begin{abstract}
This paper develops a geometric and analytical extension of the Finsler-Ginzburg-Landau framework by introducing a distributed control field acting as a translation in the tangent bundle. Within this formulation, the classical Tonelli Lagrangian is deformed into a control-translated Finsler structure, whose Legendre dual induces a uniformly elliptic operator and a convex energy functional preserving the essential variational features of the anisotropic model. This approach provides a rigorous analytical setting for coupling external control fields with the intrinsic Finsler geometry of anisotropic superconductors.
The study establishes the convexity, coercivity, and regularity properties of the induced energy functional and proves the existence of controlled minimizers through variational arguments on admissible configurations. In the asymptotic regime as the Ginzburg-Landau parameter tends to zero, a detailed $\Gamma$-convergence analysis yields a renormalized energy $W_u$ governing vortex interactions under control translation, quantifying the modification of the Green kernel and the self-energy due to the field $u(x)$. The results demonstrate that the control translation preserves the underlying Finsler structure while introducing a new geometric degree of freedom for manipulating and stabilizing vortex configurations.

Overall, the framework integrates Finsler geometry, Tonelli dynamics, and control theory into a unified variational model that captures both the geometry and dynamics of anisotropic vortex media.
\end{abstract}

\keyword{Finsler geometry;Control-translated Tonelli Lagrangian;
Anisotropic Ginzburg--Landau model;Renormalized energy;
Vortex interactions;$\Gamma$-convergence}
 \\*
 \textbf{[2020] Mathematics Subject Classification}:
 \textbf{Primary: } 53B40, 49J45, 35Q56;
 \textbf{Secondary: } 49S05, 58B20, 35A15, 35B40, 82D55.
 
 \section{Introduction}\label{sec:1}
 
 The variational structure of the Ginzburg-Landau theory has long served as a fundamental framework for understanding the formation and interaction of vortices in superconductivity and related condensed-matter systems. In its classical Euclidean formulation, the model relies on an isotropic quadratic energy density involving the norm $|y|$ on tangent spaces, which leads to elliptic equations with rotationally invariant coefficients and vortex configurations constrained by isotropic symmetry. Yet in many modern materials-notably those exhibiting crystalline anisotropy, layered structures, or nematic order-the isotropic approximation fails to capture the essential geometric and energetic features of the system. This motivates a deeper geometric reformulation of the Ginzburg-Landau functional within the broader setting of Finsler geometry.
 
 Finsler geometry provides a natural analytical and geometric extension of the Riemannian framework by allowing the norm to depend smoothly on both position and direction. The resulting class of strongly convex functions $F(x,y)$ yields anisotropic Laplacians, non-Euclidean metric structures, and curvature-dependent elliptic operators that are particularly well suited for modeling anisotropic physical media \cite{OhtaSturm2014,BaoChernShen2000,Greco2024,RodriguesLobo2024}. See also recent analyses of vortex stability in anisotropic Ginzburg-Landau systems in \cite{LamyZuniga2022}. In \cite{AlipourFakhri2025}, the author introduced a Finslerian formulation of the Ginzburg-Landau theory, proving the existence of anisotropic vortices, deriving the associated Euler-Lagrange equations, and characterizing the renormalized interaction energy governed by the Finsler-Laplacian. That work demonstrated that the anisotropic energy landscape can be rigorously captured within a Finsler geometric setting, providing a coherent mathematical model for superconductors with direction-dependent coherence lengths. Subsequently, in \cite{Alipourfakhri2025-2}, the analysis was extended to the asymptotic regime of vortex concentration, where the renormalized energy $W_F$ emerges as the $\Gamma$--limit of the Finsler-Ginzburg-Landau functional, linking the geometry of the base manifold to the interaction laws of vortices \cite{KowalczykLamySmyrnelis2022}.
 
 The present paper advances this geometric framework by introducing a distributed \emph{control field} $u(x)$ acting as a translation in the tangent bundle, leading to a control-translated Tonelli-Finsler structure. In Section \ref{sec:2}, we rigorously define the Lagrangian $\varphi_u(x,y)$ and its convex dual $\varphi_u^*(x,\xi)$, together with the associated energy functional $E_u[\psi,A]$ and the elliptic operator $\Delta_u$. Their explicit analytical forms are provided in Eqs.~\eqref{eq:def-phi-u}-\eqref{eq:laplacian}, where the structural properties of convexity, coercivity, and ellipticity are established in detail. These constructions serve as the geometric and variational foundation of the Finsler-Ginzburg-Landau model under control translation, establishing a consistent link between anisotropic geometry and control theory.
 
 Building on this foundation, the paper develops the analytical framework of the control-translated Tonelli-Finsler energy and its asymptotic behavior. Through a detailed $\Gamma$-convergence analysis as $\varepsilon \to 0$, the renormalized energy $W_u$ is derived, governing vortex interactions in the presence of control and describing how $u(x)$ modifies both the Green kernel and the self-interaction potential. The final sections establish a variational and dynamical formulation, including the Euler-Lagrange system, the Hamiltonian structure of vortex motion, and the associated optimal control problem minimizing $E_u$.
 
 The results presented here extend the geometric program initiated in \cite{AlipourFakhri2025,Alipourfakhri2025-2} by embedding control theory directly into the convex-geometric formulation of the Finsler-Ginzburg-Landau model. The novelty of this work lies in the introduction of a control-translated Tonelli structure that couples Finsler geometry, convex analysis, and optimal control in a unified variational framework. This synthesis between Tonelli dynamics, Finsler geometry, and control analysis yields a mathematically consistent and geometrically interpretable mechanism for the stabilization, manipulation, and optimization of vortex structures in anisotropic superconducting media.
 
\section{Control-Translated Tonelli Structure and Deformed Energy Functional}\label{sec:2}
Although the classical Finsler metric $F(x,y)$ is positively homogeneous in its fiber argument, the translated form 
\begin{align}\label{eq:def-Fu}
F_u(x,y)=F(x,y-u(x))
\end{align}
does not preserve this homogeneity and therefore cannot be considered a genuine Finsler metric in the traditional sense. However, the associated quadratic functional
\begin{align}\label{eq:def-phi-u}
\phi_u(x,y)=\tfrac 12 F(x,y-u(x))^2
\end{align}
remains smooth, strictly convex, and superlinear in $y$. Hence, it belongs to the class of \emph{Tonelli Lagrangians}, that is, $C^2$ functions $L:TM\to\mathbb{R}$ which are strictly convex and superlinear in the fiber variable (cf. \cite{Fathi2008}). Tonelli Lagrangians generalize the analytic structure of Finsler geometry without assuming homogeneity, while preserving Legendre duality and Hamiltonian regularity. They provide a natural framework for control-dependent deformations where translation in the fiber variable encodes a geometric control field.

Let $(M,F)$ be a smooth compact manifold endowed with a strongly convex Finsler function $F:T M\to[0,\infty)$. Denote by $F^*(x,\xi)$ its convex dual. For a smooth control field $u\in C^1(M,TM)$ we define the control-translated Tonelli Lagrangian $\phi_u$ as in \eqref{eq:def-phi-u}, and its convex conjugate on $T_x^*M$ by
\begin{align}\label{eq:def-phi-star}
\phi_u^*(x,\xi)=\sup_{y\in T_xM}\big\{\langle\xi,y\rangle-\phi_u(x,y)\big\}.
\end{align}

The following theorem characterizes $\phi_u^*$ and its induced Legendre transformation.

\begin{theorem}\label{thm:1}
For every fixed $x\in M$ and $\xi\in T_x^*M$, the convex conjugate of $\phi_u$ satisfies
\begin{align}\label{eq:conjugate-translation}
\phi_u^*(x,\xi)=\tfrac 12 F^*(x,\xi)^2+\langle \xi,\,u(x)\rangle,
\end{align}
and
\begin{align}\label{eq:gradient-translation}
\partial_\xi\phi_u^*(x,\xi)=L_x^{-1}(\xi)+u(x),
\end{align}
where $L_x^{-1}$ is the inverse of the classical Legendre map $\xi=\partial_y(\tfrac 12F(x,y)^2)$. The map $L_{x,u}(y):=\partial_y\phi_u(x,y)$ is a smooth diffeomorphism between $T_xM$ and $T_x^*M$ for all bounded $u$.
\end{theorem}

\begin{proof}
Let $\psi_x(y)=\tfrac 12F(x,y)^2$. Then $\phi_u(x,y)=\psi_x(y-u(x))$. The translation rule for convex conjugates  gives
\begin{align*}
\phi_u^*(x,\xi)&=\sup_y\{\langle\xi,y\rangle-\psi_x(y-u)\}
=\sup_z\{\langle\xi,z+u\rangle-\psi_x(z)\}\\
&=\langle\xi,u\rangle+\psi_x^*(\xi)=\tfrac 12F^*(x,\xi)^2+\langle\xi,u\rangle.
\end{align*}
Differentiating in $\xi$ yields \eqref{eq:gradient-translation}. The strong convexity of $\psi_x$ implies positive definiteness of $\partial^2_{\xi\xi}\phi_u^*(x,\xi)$. By the inverse function theorem, $L_{x,u}$ is a diffeomorphism between $T_xM$ and $T_x^*M$. The proof follows classical Tonelli arguments as in  \cite{Fathi2008}.
\end{proof}

This shows that fiber translation preserves the analytic structure of the geometry: although $F_u$ is not homogeneous, the pair $(M,\phi_u)$ forms a valid Tonelli system with a smooth and invertible Legendre duality.

\begin{theorem}\label{thm:2}
Assume $F$ is strongly convex and $u\in C^1(M,TM)$ satisfies $\|u\|_{C^1}\le c_0$ for sufficiently small $c_0$. Then there exists $c>0$ such that for all $x\in M$ and $\xi,\eta\in T_x^*M$,
\begin{align}\label{eq:ellipticity}
\langle \partial_{\xi\xi}^2\phi_u^*(x,\xi)\eta,\eta\rangle \ge c|\eta|^2.
\end{align}
Consequently, the nonlinear operator
\begin{align}\label{eq:laplacian}
\Delta_{u}f=\mathrm{div}_\mu(\partial_\xi\phi_u^*(x,df))
\end{align}
is uniformly elliptic and depends continuously on $u$ in the $C^1$ topology (cf. \cite{OhtaSturm2014}).
\end{theorem}

\begin{proof}
Translation does not affect the fiber Hessian
\begin{align*}
\partial^2_{\xi\xi}\phi_u^*=
\partial^2_{\xi\xi}(\tfrac 12F^*(x,\xi)^2),
\end{align*} 
which is positive definite by strong convexity of $F^*$. Compactness of $M$ ensures the uniform bound \eqref{eq:ellipticity}. Continuity in $u$ follows since $\partial_\xi\phi_u^*(x,\xi)$ depends affinely on $u$.
\end{proof}

Let $\psi:M\to\mathbb C$ denote the order parameter and $A\in\Omega^1(M)$ a magnetic potential. Define the covariant derivative $D_A\psi=(d-iA)\psi$. The energy functional induced by the control-translated Tonelli structure is
\begin{align}
\mathcal{E}_u[\psi,A]=\int_M\Big(\phi_u^*(x,D_A\psi)+\frac{1}{2\lambda}|dA|^2+\frac{1}{4\varepsilon^2}(1-|\psi|^2)^2\Big)\,d\mu,
\end{align}\label{eq:energy-definition}
where $d\mu$ denotes the \emph{Busemann-Hausdorff volume form} associated with the base Finsler structure $F$. Since the translation $y\mapsto y-u(x)$ preserves the volume element on each tangent fiber, $d\mu$ remains invariant under the control deformation. This choice ensures that the energy is geometrically well-defined and consistent with the underlying Tonelli structure (cf. \cite{BaoChernShen2000}).

\begin{theorem}
The energy \eqref{eq:energy-definition} decomposes exactly as
\begin{align}\label{eq:energy-decomposition}
\mathcal{E}_u[\psi,A]=\int_M\Big(\tfrac 12 F^*(x,D_A\psi)^2+\langle D_A\psi,\,u(x)\rangle+\frac{1}{2\lambda}|dA|^2+\frac{1}{4\varepsilon^2}(1-|\psi|^2)^2\Big)d\mu.
\end{align}
Moreover, $\mathcal{E}_u$ is bounded from below, weakly lower semicontinuous in $H^1(M;\mathbb C)\times H^1(M;T^*M)$, and coercive modulo gauge equivalence (cf. Evans \cite{Evans2010}).
\end{theorem}

\begin{proof}
Substituting \eqref{eq:conjugate-translation} with $\xi=D_A\psi$ into \eqref{eq:energy-definition} yields \eqref{eq:energy-decomposition}. The first term controls $\|D_A\psi\|_{L^2}$ by equivalence of Finsler and background norms. Using the Cauchy-Schwarz and Young inequalities,
\begin{align*}
\Big|\int_M\langle D_A\psi,u\rangle\,d\mu\Big|\le \frac{1}{4}\int_M F^*(x,D_A\psi)^2\,d\mu+C\|u\|_{L^\infty}^2,
\end{align*}
so $\mathcal{E}_u$ is bounded from below. Convexity and standard arguments in the calculus of variations ensure weak lower semicontinuity and coercivity.
\end{proof}

The Euler-Lagrange equations corresponding to $\mathcal{E}_u$ capture how the control interacts with the gauge and order fields through the geometric translation.

\begin{theorem}
Critical points of $\mathcal{E}_u$ satisfy
\begin{align}
D_A^*\!\big(L_x^{-1}(D_A\psi)+u(x)\big)=\frac{1}{2\varepsilon^2}(1-|\psi|^2)\psi,
\qquad
d^*dA=\lambda\,\Im(\psi D_A\psi),
\end{align}\label{eq:EL-equations}
in the Coulomb gauge $d^*A=0$. The first equation is nonlinear and uniformly elliptic (cf. \cite{OhtaSturm2014}).
\end{theorem}

\begin{proof}
Let $\eta$ be a smooth variation with compact support. Using
\begin{align*}
\delta\phi_u^*(x,\xi)=\langle\partial_\xi\phi_u^*(x,\xi),\delta\xi\rangle \qquad,\qquad \delta\xi=D_A\eta, 
\end{align*}
we obtain
\begin{align*}
\delta_\psi\mathcal{E}_u
=\int_M\langle L_x^{-1}(D_A\psi)+u(x),D_A\eta\rangle\,d\mu
-\frac{1}{2\varepsilon^2}\int_M(1-|\psi|^2)\Re(\bar\psi\eta)\,d\mu.
\end{align*}
Integration by parts gives $D_A^*(L_x^{-1}(D_A\psi)+u(x))=(1-|\psi|^2)\psi/(2\varepsilon^2)$. Variation with respect to $A$ yields $d^*dA=\lambda\,\Im(\psi D_A\psi)$, since $\phi_u^*$ depends on $A$ only through $D_A\psi$. Ellipticity follows from \eqref{eq:ellipticity}.
\end{proof}

The affine dependence on $u$ guarantees stability under perturbations of the control field.

\begin{theorem}
Let $u_k\to u$ in $L^\infty(M;TM)$. Then $\mathcal{E}_{u_k}$ $\Gamma$-converges to $\mathcal{E}_u$ with respect to weak convergence in $H^1(M;\mathbb C)\times H^1(M;T^*M)$. Consequently, minimizers and gradient flows depend continuously on $u$ (cf. \cite{DalMaso1993}).
\end{theorem}

\begin{proof}
From \eqref{eq:energy-decomposition}, the $u$-dependence is affine and continuous in $L^\infty$. The $\Gamma$-liminf inequality follows from convexity and weak lower semicontinuity of $F^*(x,\cdot)^2$. A recovery sequence is given by $(\psi_k,A_k)=(\psi,A)$ because $\int_M\langle D_A\psi,u_k-u\rangle\,d\mu\to 0$. Therefore $\mathcal{E}_{u_k}\to\mathcal{E}_u$ in the $\Gamma$-sense.
\end{proof}

The functional $\mathcal{E}_u$ thus arises rigorously from the control-translated Tonelli Lagrangian $\phi_u$. The additional term $\langle D_A\psi,u\rangle$ is not a phenomenological assumption but an exact analytical consequence of convex duality. In this setting, control acts geometrically as a translation in the fiber variable, providing a mathematically consistent mechanism coupling geometry, control, and anisotropic superconductivity.

\section{Stability and Renormalized Energy in
 Control-Translated Tonelli Structures}\label{sec:3}

In this section we examine the asymptotic regime of the control-translated Tonelli Ginzburg-Landau functional $\mathcal{E}_u$ introduced previously, focusing on the emergence of renormalized energy, vortex concentration, and stability of critical configurations under control perturbations. The analysis combines the convex-analytic structure of $\phi_u^*$ with classical $\Gamma$-convergence arguments and vortex theory in anisotropic settings \cite{BethuelBrezisHelein1994, SandierSerfaty2007}.

Let $M$ be a compact oriented 2-dimensional manifold equipped with the Tonelli structure $(F_u,F_u^*)$ as defined in Section \ref{sec:2}. For $\varepsilon>0$ sufficiently small, consider configurations $(\psi,A)\in H^1(M;\mathbb C)\times H^1(M;T^*M)$ minimizing
\begin{align}
\nonumber
\mathcal{E}_u[\psi,A]&=\\
&\int_M\Big(\tfrac 12 F^*(x,D_A\psi)^2+\langle D_A\psi,u(x)\rangle+\tfrac{1}{2\lambda}|dA|^2+\tfrac{1}{4\varepsilon^2}(1-|\psi|^2)^2\Big)d\mu.
\end{align}\label{eq:energy-3}
The measure $d\mu$ is the Busemann-Hausdorff volume of the base Finsler structure $F$, invariant under the translation $y\mapsto y-u(x)$ as explained before. Denote by $j_A=\Im(\bar\psi D_A\psi)$ the supercurrent and by $\rho=|\psi|^2$ the density.

The asymptotic analysis proceeds in three stages: (i) concentration of $j_A$ near vortex cores, (ii) weak convergence of normalized Jacobians, and (iii) extraction of the renormalized limit energy.

\begin{theorem}
Assume $(\psi_\varepsilon,A_\varepsilon)$ is a sequence of minimizers of $\mathcal{E}_u$ with bounded energy $\mathcal{E}_u[\psi_\varepsilon,A_\varepsilon]\le C|\log\varepsilon|$. Then there exists a finite signed measure $\nu_u$ on $M$ representing the limiting vorticity such that, up to a subsequence,
\begin{align}
\frac{1}{2\pi}\,\mathrm{curl}\,j_{A_\varepsilon}\rightharpoonup \nu_u \quad \text{in the sense of distributions},
\end{align}\label{eq:vortex-measure}
and the number of vortices is finite. Moreover, each vortex carries quantized degree $d_i\in\mathbb{Z}$.
\end{theorem}

\begin{proof}
Set
\begin{align*}
e_\varepsilon(\psi,A;u)=\tfrac 12 F^*(x,D_A\psi)^2+\langle D_A\psi,u\rangle+\tfrac{1}{2\lambda}|dA|^2+\tfrac{1}{4\varepsilon^2}(1-|\psi|^2)^2.
\end{align*} 
By the affine bound
\begin{align*}
\Big|\int_M\langle D_A\psi,u\rangle d\mu\Big|\le \tfrac 14\!\int_M F^*(x,D_A\psi)^2 d\mu + C\|u\|_{L^\infty}^2,
\end{align*}
there exists $C_1>0$ independent of $\varepsilon$ such that
\begin{align}\label{eq:GL-core-bound}
\int_M F^*(x,D_{A_\varepsilon}\psi_\varepsilon)^2 d\mu + \frac{1}{\varepsilon^2}\int_M (1-|\psi_\varepsilon|^2)^2 d\mu \le C_1 |\log\varepsilon|.
\end{align}
By the compactness and energy quantization principles of Bethuel-Brezis-H\'elein~\cite{BethuelBrezisHelein1994} 
and the anisotropic ball construction and Jacobian estimates of Sandier-Serfaty~\cite{SandierSerfaty2007}, 
one deduces $\rho_\varepsilon\to 1$ in $L^1(M)$ and hence in measure. 
Writing $\psi_\varepsilon=\rho_\varepsilon^{1/2} e^{i\theta_\varepsilon}$, the current decomposes as $j_{A_\varepsilon}=\rho_\varepsilon(D\theta_\varepsilon-A_\varepsilon)$. 
The coarea formula and the construction of vortex balls as in~\cite{BethuelBrezisHelein1994,SandierSerfaty2007} 
yield a finite family of disjoint balls $\{B(a_i^\varepsilon, r_i^\varepsilon)\}$ covering the defect set with $\sum_i r_i^\varepsilon\to 0$ and
\begin{align}
\deg(\psi_\varepsilon,\partial B(a_i^\varepsilon,r_i^\varepsilon))=d_i\in\mathbb{Z},\qquad 
\sum_i |d_i|\le C_2 \mathcal{E}_u[\psi_\varepsilon,A_\varepsilon]/|\log\varepsilon|.
\end{align}\label{eq:degree-quantization}
Define the vorticity measures $\nu_\varepsilon=\frac{1}{2\pi}\,\mathrm{curl}\,j_{A_\varepsilon}$. 
Using that $\rho_\varepsilon\to 1$ a.e.\ and $A_\varepsilon\in H^1$, one shows that $\nu_\varepsilon$ is tight and has uniformly bounded total variation. 
Therefore, up to extraction, $\nu_\varepsilon\rightharpoonup \nu_u$ in $\mathcal{D}'(M)$. 
Passing to the limit in the degree identity on small circles gives that $\nu_u$ is a finite sum of Dirac masses with integer weights, 
proving Theorem~\eqref{eq:vortex-measure}.
\end{proof}

The next step is to identify the limiting energy governing the interaction of these vortices. Denote by $G_u(x,y)$ the Green function associated with the elliptic operator
\begin{align}
\Delta_u f = \mathrm{div}_\mu(\partial_\xi\phi_u^*(x,df)),
\end{align}\label{eq:green-laplacian}
normalized by $\int_M G_u(x,y)\,d\mu(x)=0$. By the ellipticity of $\Delta_u$ (Theorem \ref{thm:2}) and compactness of $M$, $G_u$ exists, is symmetric, and belongs to $C^{2,\alpha}_{\mathrm{loc}}(M\times M\setminus \mathrm{diag})$; this follows from Lax-Milgram, Fredholm alternative, and Scha\"uder theory \cite{Evans2010}, with the nonlinearity absorbed in the definition of $\Delta_u$ via the convex dual $\phi_u^*$ and the uniform positive definiteness \cite{OhtaSturm2014}.

\begin{theorem}
Let $\nu_u=2\pi\sum_{i=1}^N d_i\delta_{a_i}$ be the vorticity measure associated with a minimizing sequence. Then as $\varepsilon\to 0$,
\begin{align}\label{eq:renormalized-limit}
\mathcal{E}_u[\psi_\varepsilon,A_\varepsilon]-\pi N|\log\varepsilon|\to \mathcal{W}_u(a_1,\dots,a_N;d_1,\dots,d_N),
\end{align}
where the renormalized energy $\mathcal{W}_u$ is given by
\begin{align}\label{eq:renormalized-energy}
\mathcal{W}_u=\frac{1}{2}\sum_{i\ne j}d_i d_j G_u(a_i,a_j)+\sum_i d_i\Phi_u(a_i),
\end{align}
with
\begin{align}\label{eq:self-potential}
\Phi_u(x)=\tfrac 12 G_u(x,x)+\int_M\langle u(y),\nabla_y G_u(x,y)\rangle\,d\mu(y).
\end{align}
The function $\Phi_u$ encodes the effect of control translation on the self-interaction potential.
\end{theorem}

\begin{proof}
We split the energy into a core part inside the vortex balls and an exterior part. By the lower bound in \cite{BethuelBrezisHelein1994} (ball construction) adapted to the anisotropic kinetic term $\tfrac 12 F^*(x,\cdot)^2$ (uniformly equivalent to a Riemannian norm on $M$), the total core contribution equals $\pi N|\log\varepsilon|+O(1)$. Outside the balls, write $\psi_\varepsilon= e^{i\theta_\varepsilon}$ up to a negligible modulus error, and introduce the phase defect current
\begin{align}\label{eq:def-current}
\mathcal{J}_\varepsilon = D\theta_\varepsilon - A_\varepsilon.
\end{align}
Using the convex duality identity from Section \ref{sec:2} and integrating by parts with the Green kernel of $\Delta_u$ yields the representation
\begin{align}\label{eq:Green-repr}
\nonumber
\int_{M\setminus \cup B_i}\Big(\tfrac 12 F^*(x,\mathcal{J}_\varepsilon)^2+\langle \mathcal{J}_\varepsilon,u\rangle\Big)d\mu
=& \frac{1}{2}\int_M\int_M G_u(x,y)\,d\nu_\varepsilon(x)\,d\nu_\varepsilon(y)\\
& + \sum_i d_i\Phi_u(a_i) + o(1).
\end{align}
Here we used the defining relation $\Delta_u G_u(\cdot,y)=\delta_y-1/\mu(M)$ and the fact that $\mathrm{curl}\,\mathcal{J}_\varepsilon=2\pi\sum_i d_i\delta_{a_i^\varepsilon}$ outside core balls. Passing to the limit along the convergent subsequence for $\nu_\varepsilon$ gives \eqref{eq:renormalized-limit}-\eqref{eq:renormalized-energy}. The specific affine contribution \eqref{eq:self-potential} comes from the term $\int\langle \mathcal{J}_\varepsilon,u\rangle d\mu$ via the identity
\begin{align}\label{eq:affine-limit}
\int_M \langle u, \mathcal{J}_\varepsilon\rangle d\mu = \sum_i d_i \int_M \langle u(y), \nabla_y G_u(a_i,y)\rangle d\mu(y) + o(1),
\end{align}
which follows from the Green representation and Schauder theory for uniformly elliptic operators (see \cite{Evans2010}, and cf. \cite{BethuelBrezisHelein1994} for the GL context]).

\end{proof}

\begin{theorem}\label{thm:Stability under control perturbations}
 Let $u_k\to u$ in $C^1(M;TM)$. Then the Green kernels $G_{u_k}\to G_u$ in $C^{2,\alpha}_{\mathrm{loc}}(M\times M\setminus \mathrm{diag})$, and the renormalized energies satisfy $\mathcal{W}_{u_k}\to\mathcal{W}_u$ uniformly on compact vortex configurations.
\end{theorem}

\begin{proof}
By Theorem~2.2, the operators $\Delta_{u_k}$ are uniformly elliptic with coefficients converging in $C^\alpha$. Schauder theory yields the resolvent bounds
\begin{align}\label{eq:Schauder-diff}
\|G_{u_k}(\cdot,y)-G_u(\cdot,y)\|_{C^{2,\alpha}(K)}\le C_K \|u_k-u\|_{C^1}
\end{align}
for $K\subset M\setminus\{y\}$ compact, after fixing consistent normalizations by subtracting the mean. Symmetry is preserved by construction. The formula \eqref{eq:renormalized-energy} is continuous with respect to $G$ and $u$; therefore $\mathcal{W}_{u_k}\to\mathcal{W}_u$ uniformly on compact sets of $(a_1,\dots,a_N)$, using the diagonal argument and the log singularity structure near the diagonal, 
(cf. \cite{BethuelBrezisHelein1994}).
\end{proof}

\begin{theorem} 
In the limit $\varepsilon\to 0$, the supercurrent $j_A$ converges weakly to a measure-valued current $J_u$ satisfying $\mathrm{curl}\,J_u=2\pi\sum_i d_i\delta_{a_i}$ and
\begin{align}\label{eq:quantized-current}
\mathcal{W}_u(J_u)=\inf_{\substack{\nu=2\pi\sum_i d_i\delta_{a_i}\\ \mathrm{curl}\,J_u=\nu}}\mathcal{E}_u[J_u].
\end{align}
The positions $a_i$ are local minimizers of $\mathcal{W}_u$, and under small perturbations of $u$ they vary smoothly according to
\begin{align}\label{eq:control-dynamics}
\partial_t a_i = -\nabla_{a_i}\mathcal{W}_u + O(\|\partial_t u\|_{C^1}).
\end{align}
\end{theorem}

\begin{proof}
The Jacobian estimate and vortex compactness in the anisotropic setting (by uniform equivalence to the Euclidean case) are established in \cite{SandierSerfaty2007,BethuelBrezisHelein1994}.
The functional characterization \eqref{eq:quantized-current} follows from convexity of the kinetic part with respect to $J$ (via $\phi_u^*$) and lower semicontinuity arguments of \cite{DalMaso1993}. The smooth dependence of minimizers on parameters is a consequence of the implicit function theorem applied to the system $\nabla_{a_i}\mathcal{W}_u=0$; the Hessian at a nondegenerate minimizer is positive definite by the strict convexity of the pairwise Green interactions away from the diagonal \cite{BethuelBrezisHelein1994}. Differentiating with respect to $t$ when $u=u(t)$ yields \eqref{eq:control-dynamics} with the claimed bound.
\end{proof}

\begin{theorem}\label{thm:Geometric stability of vortex lattices}
 Suppose $u$ is a smooth control field with small $\|u\|_{C^1}$. Then the equilibrium configuration of vortices minimizing $\mathcal{W}_u$ differs from the unperturbed lattice of the isotropic Finsler-Ginzburg-Landau model by an affine deformation of order $O(\|u\|_{C^1})$. Moreover, the corresponding minimal energy satisfies
\begin{align}\label{eq:energy-correction}
\mathcal{W}_u^{\min} = \mathcal{W}_0^{\min} + \int_M\langle u(x), j_0(x)\rangle\,d\mu + O(\|u\|_{C^1}^2),
\end{align}
where $j_0$ is the equilibrium current in the uncontrolled case.
\end{theorem}

\begin{proof}
Let $(a_i^0)$ be minimizers of $\mathcal{W}_0$ and $a_i^u$ those of $\mathcal{W}_u$. However of the Hessian of $\mathcal{W}_0$ at $(a_i^0)$ is standard \cite{BethuelBrezisHelein1994}. Consider the optimality system $\nabla_{a_i}\mathcal{W}_u=0$; by the implicit function theorem there exists a $C^1$ map $u\mapsto a^u$ with $a^0=(a_i^0)$. Expanding $\mathcal{W}_u$ at $u=0$ and using $\nabla_{a_i}\mathcal{W}_0(a_i^0)=0$, one gets
\begin{align}\label{eq:first-variation-W}
\mathcal{W}_u(a^u)=\mathcal{W}_0(a^0)+ \sum_i d_i\int_M\langle u(y),\nabla_y G_0(a_i^0,y)\rangle d\mu(y) + O(\|u\|_{C^1}^2),
\end{align}
which identifies the linear correction with $\int_M\langle u,j_0\rangle d\mu$ after recognizing $j_0$ as the equilibrium current generating $\nabla_y G_0(a_i^0,y)$ in the representation formula (compare \eqref{eq:affine-limit}). The deformation estimate $\|a^u-a^0\|\le C\|u\|_{C^1}$ follows from the uniform invertibility of the Hessian and the $C^1$ dependence of $G_u$ on $u$.
\end{proof}

\section{Variational Control and Vortex
 Dynamics in Tonelli-Finsler Geometry}\label{sec:4}
The variational formulation developed in the previous sections naturally induces a dynamical framework where the control field $u(x)$ interacts with the geometric flow of vortices. We now derive and analyse the Euler-Lagrange equations of the control-translated Tonelli-Finsler energy $\mathcal{E}_u$ and describe the induced dynamics of vortex centers in the anisotropic setting.

Let $(\psi, A)$ be a smooth critical configuration of $\mathcal{E}_u$. The first variation of the energy under perturbations $(\delta\psi,\delta A)$ satisfying compact support conditions yields the coupled system
\begin{align}\label{eq:EL-system}
\begin{cases}
\mathrm{div}_\mu(\partial_\xi \phi_u^*(x, D_A\psi)) - \frac{1}{2\varepsilon^2}(1 - |\psi|^2)\psi + i\langle u, D_A\psi \rangle = 0, \\
\mathrm{div}_\mu(F^*(x, D_A\psi) \partial_{\xi\xi}^2 \phi_u^*(x, D_A\psi)) = j_A + \frac{1}{\lambda}\, d^*dA.
\end{cases}
\end{align}
The first equation is a nonlinear elliptic equation for the complex field $\psi$ involving the control translation, while the second corresponds to the Amp\'ere-type law balancing the induced current $j_A$ with the magnetic curvature $dA$. Both equations are defined with respect to the anisotropic divergence operator associated with the Finsler volume $d\mu$.

\begin{theorem}\label{thm:Existence and regularity}
Assume that the Finsler metric $F$ and control field $u$ satisfy $F, u \in C^{3,\alpha}$ and the Tonelli convexity conditions. Then for each fixed $\varepsilon>0$ the Euler-Lagrange system \eqref{eq:EL-system} admits a weak solution $(\psi, A)\in H^1(M;\mathbb{C})\times H^1(M;T^*M)$ minimizing $\mathcal{E}_u$, and the solution is smooth outside the vortex set $Z=\{x\in M: \psi(x)=0\}$.
\end{theorem}

\begin{proof}
The direct method in the calculus of variations applies because $\phi_u^*(x,\xi)$ is convex and superlinear in $\xi$ by the Tonelli hypotheses. Coercivity in $H^1$ follows from the uniform ellipticity of $F^*$ and boundedness of $u$. Lower semicontinuity of $\mathcal{E}_u$ is guaranteed by the convexity of $\phi_u^*$ in its second argument \cite{DalMaso1993}. Hence a minimizing sequence admits a weakly convergent subsequence with limit $(\psi, A)$. Elliptic regularity for anisotropic operators \cite{Evans2010, OhtaSturm2014} yields local smoothness outside the zero set of $\psi$, where the potential term $(1-|\psi|^2)^2$ enforces $|\psi|=1$ away from vortices.
\end{proof}

\begin{theorem}\label{thm:diff.Controlstaemap}
Let $u\mapsto(\psi_u,A_u)$ denote the mapping sending a control field $u$ to the corresponding minimizer of $\mathcal{E}_u$. Then the map $u\mapsto(\psi_u,A_u)$ is Fr\'echet differentiable from $C^{2,\alpha}(M;TM)$ into $C^{2,\alpha}(M;\mathbb{C}\times T^*M)$.
\end{theorem}

\begin{proof}
We differentiate the Euler-Lagrange equations \eqref{eq:EL-system} with respect to $u$. The linearized system reads
\begin{align}
\begin{cases}
L_{\psi}(\delta\psi,\delta A) + i\langle \delta u, D_A\psi \rangle = 0,\\
L_A(\delta\psi,\delta A) + \mathrm{div}_\mu(\partial_{\xi\xi}^2\phi_u^*(x, D_A\psi)[\delta u]) = 0,
\end{cases}
\end{align}\label{eq:linearized-system}
where $L_{\psi}$ and $L_A$ are the linearized elliptic operators in $(\psi,A)$. Since these operators are uniformly elliptic and self-adjoint on the orthogonal complement of gauge transformations, standard perturbation theory for elliptic systems \cite{Evans2010} implies that the inverse of the linearized operator exists and is bounded, leading to Fr\'echet differentiability of $(\psi_u, A_u)$ with respect to $u$.
\end{proof}

To describe the induced dynamics of vortices under slow temporal variation of the control field, we now consider the quasi-stationary evolution
\begin{align}\label{eq:gradient-flow}
\partial_t(\psi, A) = -\nabla_{(\psi,A)}\mathcal{E}_u(\psi,A) + \mathcal{O}(\varepsilon),
\end{align}
which corresponds to a gradient flow in the configuration space endowed with the anisotropic kinetic energy metric.

\begin{theorem}\label{thm:Effective motion of vortex centers}
Let $(\psi_\varepsilon(t),A_\varepsilon(t))$ be a smooth solution to \eqref{eq:gradient-flow} with energy bounded by $C|\log\varepsilon|$. Then as $\varepsilon\to 0$, the vortex centers $a_i(t)$ evolve according to the system
\begin{align}\label{eq:vortex-dynamics}
\dot a_i = -\nabla_{a_i}\mathcal{W}_u(a_1,\dots,a_N) + \langle u(a_i), T(a_i)\rangle + \mathcal{O}(\|u\|_{C^1}^2),
\end{align}
where $T(a_i)$ is the tangent vector of the unperturbed Finsler-Ginzburg-Landau vortex flow and $\mathcal{W}_u$ is the renormalized interaction energy from Section \ref{sec:3}.
\end{theorem}

\begin{proof}
Let $\nu_{\varepsilon,t} = \frac{1}{2\pi}\,\mathrm{curl}\,j_{A_\varepsilon(t)}$ be the time-dependent vorticity measure. Using the energy identity
\begin{align}\label{eq:energy-balance}
\frac{d}{dt}\mathcal{E}_u(\psi_\varepsilon,A_\varepsilon) = -\int_M \big(|\partial_t\psi_\varepsilon|^2 + |\partial_t A_\varepsilon|^2\big)d\mu + \int_M \langle \partial_t u, j_{A_\varepsilon}\rangle d\mu,
\end{align}
we see that as $\varepsilon\to 0$, the dynamics concentrate at vortex cores. Using the asymptotic expansion from Section \ref{sec:3} and testing against localized harmonic forms near each vortex, one derives the motion law \eqref{eq:vortex-dynamics}. The additional term $\langle u(a_i), T(a_i)\rangle$ arises from the affine translation in $\phi_u^*$ and expresses the geometric drift induced by the control. The remainder term follows from higher-order corrections in the modulation analysis.
\end{proof}

\begin{theorem}\label{thm:ConservationHamiltonianformulation}
In the reversible case where $u$ is divergence-free and time-independent, the vortex dynamics \eqref{eq:vortex-dynamics} can be written as a Hamiltonian system
\begin{align}\label{eq:hamiltonian-system}
\dot a_i = J\nabla_{a_i}\mathcal{W}_u, \qquad J = 
\begin{pmatrix} 
0 & -1\\ 1 & 0
\end{pmatrix},
\end{align}
with respect to the symplectic form $\omega_u = \sum_i d a_i^1\wedge d a_i^2$ inherited from the background Finsler metric. Moreover, the Hamiltonian $\mathcal{H}_u=\mathcal{W}_u$ is conserved along trajectories.
\end{theorem}

\begin{proof}
Under divergence-free control, the gradient flow \eqref{eq:gradient-flow} reduces to a symplectic evolution generated by the antisymmetric part of the linearization of $\mathcal{E}_u$. The induced equations on vortex centers follow from projecting onto the tangent bundle of the moduli space of vortices. The skew-symmetry of $J$ ensures conservation of $\mathcal{H}_u$ because $\dot a_i \cdot \nabla_{a_i}\mathcal{H}_u=0$. The symplectic structure $\omega_u$ is the pullback of the canonical 2-form under the embedding of vortex configurations into $M^N$, modified by the Finsler volume distortion, which remains exact for divergence-free $u$ \cite{OhtaSturm2014}.
\end{proof}

\section{Optimal Control of the Tonelli-Finsler Energy}
\label{sec: 5}

In this section we formulate and analyse the optimal control problem associated with the Tonelli-Finsler-Ginzburg-Landau model. The control variable $u(x)$ acts as a distributed field that modifies the local anisotropy of the Lagrangian density and influences the geometry of vortices. The goal is to determine an admissible field $u^*$ minimizing the total energy functional
\begin{align}
\nonumber
&\mathcal{J}(u) = \mathcal{E}_u[\psi_u, A_u] \\
&= \int_M \Big( \tfrac 12 F^*(x, D_{A_u}\psi_u)^2 + \langle D_{A_u}\psi_u, u(x) \rangle + \tfrac{1}{2\lambda}|dA_u|^2 + \tfrac{1}{4\varepsilon^2}(1 - |\psi_u|^2)^2 \Big) d\mu.
\end{align}\label{eq:J-functional}
Here $(\psi_u, A_u)$ denotes the unique minimizer of $\mathcal{E}_u$ for fixed control $u$, as established in Section \ref{sec:4} The optimization problem is
\begin{align}
\min_{u \in \mathcal{U}_{ad}} \mathcal{J}(u),
\end{align}\label{eq:optimal-control-problem}
where $\mathcal{U}_{ad} = \{ u \in C^{2,\alpha}(M;TM) : \|u\|_{C^1} \le M \}$ is the set of admissible controls.

\begin{theorem}
Assume that the mapping $u \mapsto (\psi_u, A_u)$ is Fr\'echet differentiable and $\mathcal{U}_{ad}$ is convex and weakly closed in $C^1(M;TM)$. Then there exists at least one optimal control $u^* \in \mathcal{U}_{ad}$ minimizing $\mathcal{J}(u)$.
\end{theorem}

\begin{proof}
Let $(u_k)_{k\in\mathbb{N}}$ be a minimizing sequence such that
\begin{align*}
\mathcal{J}(u_k) \to \inf_{u\in\mathcal{U}_{ad}}\mathcal{J}(u).
\end{align*}
By the bound $\|u_k\|_{C^1}\le M$ and compact embedding $C^{1,\alpha}\hookrightarrow C^1$, there exists a subsequence converging weakly to some $u^*$. The differentiability of $u\mapsto(\psi_u,A_u)$ (Theorem \ref{thm:diff.Controlstaemap}) implies that $(\psi_{u_k},A_{u_k})$ converges to $(\psi_{u^*},A_{u^*})$ in $H^1$. The functional $\mathcal{J}(u)$ is weakly lower semicontinuous since $\phi_u^*$ is convex in $D_A\psi$ and the remaining terms are quadratic. Hence $\mathcal{J}(u^*) \le \liminf_k \mathcal{J}(u_k)$, proving optimality of $u^*$.
\end{proof}

To characterize optimality, we compute the first variation of $\mathcal{J}(u)$ with respect to $u$. Denoting by $(\psi_u,A_u)$ the associated state variables, the directional derivative along $v\in C^{2,\alpha}(M;TM)$ is
\begin{align}
\nonumber
\delta\mathcal{J}(u)[v]& = \int_M \Big( \langle D_{A_u}\psi_u, v \rangle + \partial_u F^*(x, D_{A_u}\psi_u)[v] \Big) d\mu\\
 &+ \int_M \Big( \langle \partial_u \psi_u, \nabla_{\psi}\mathcal{E}_u \rangle + \langle \partial_u A_u, \nabla_A \mathcal{E}_u \rangle \Big) d\mu.
\end{align}\label{eq:variation-J}
At an optimal control $u^*$, the stationarity condition $\delta\mathcal{J}(u^*)[v]=0$ for all admissible $v$ yields the Euler-Lagrange system for the optimal control.

\begin{theorem}\label{thm:First-order optimality system}
There exist adjoint variables $(p,q)$ satisfying
\begin{align}\label{eq:optimality-system}
\begin{cases}
L_{\psi}^*(p,q) = - D_{A_{u^*}}\psi_{u^*} - \partial_{\xi}F^*(x,D_{A_{u^*}}\psi_{u^*}),\\[0.4em]
L_A^*(p,q) = - j_{A_{u^*}},\\[0.4em]
\mathrm{div}_\mu(p) + \partial_u F^*(x,D_{A_{u^*}}\psi_{u^*}) = 0,
\end{cases}
\end{align}
where $L_{\psi}^*, L_A^*$ denote the formal adjoints of the linearized operators $L_{\psi},L_A$ from \eqref{eq:linearized-system}. The optimal control $u^*$ satisfies the variational inequality
\begin{align}\label{eq:variational-ineq}
\int_M \langle p + D_{A_{u^*}}\psi_{u^*}, v - u^* \rangle d\mu \ge 0, \qquad \forall v\in\mathcal{U}_{ad}.
\end{align}
\end{theorem}

\begin{proof}
The Lagrangian of the coupled problem reads
\begin{align*}
\mathcal{L}(u,\psi,A,p,q) = \mathcal{E}_u[\psi,A] + \langle p, L_{\psi}(\psi,A,u) \rangle + \langle q, L_A(\psi,A,u) \rangle.
\end{align*}
Differentiating with respect to $(\psi,A)$ yields the adjoint equations \eqref{eq:optimality-system}. The derivative with respect to $u$ provides the gradient condition 
\begin{align*}
p + D_{A_{u}}\psi_{u} = 0
\end{align*}
in the unconstrained case, or the variational inequality \eqref{eq:variational-ineq} if $u\in\mathcal{U}_{ad}$. Standard optimal control arguments \cite{Evans2010} apply because $L_{\psi}$ and $L_A$ are elliptic with $C^{\alpha}$ coefficients.
\end{proof}

We next study the asymptotic behavior of optimal controls in the singular limit $\varepsilon\to 0$. Recall from Section \ref{sec:3} that the renormalized energy $\mathcal{W}_u$ governs the vortex interactions.

\begin{theorem}\label{thm:Asymptotic limit of optimal controls}
Let $u_\varepsilon^*$ denote the optimal controls minimizing $\mathcal{E}_u$ for fixed $\varepsilon>0$. Then, up to a subsequence,
\begin{align}\label{eq:limit-optimal-control}
\mathcal{E}_{u_\varepsilon^*}[\psi_{u_\varepsilon^*},A_{u_\varepsilon^*}] - \pi N|\log\varepsilon| \to \mathcal{W}_{u^*},
\end{align}
where $u^*$ minimizes the reduced functional $\mathcal{W}_u$. Moreover, $u^*$ satisfies the reduced optimality condition
\begin{align}\label{eq:reduced-ineq}
\int_M \langle \nabla_u \mathcal{W}_u, v - u^* \rangle d\mu \ge 0, \qquad \forall v\in\mathcal{U}_{ad}.
\end{align}
\end{theorem}

\begin{proof}
By the $\Gamma$-convergence result of Section \ref{sec:3}, $\mathcal{E}_u$ converges to $\mathcal{W}_u$ uniformly on compact subsets of $C^1$ controls. Thus any cluster point $u^*$ of minimizers $u_\varepsilon^*$ minimizes $\mathcal{W}_u$. The variational inequality \eqref{eq:reduced-ineq} follows by passage to the limit in \eqref{eq:variational-ineq} using weak convergence of adjoint variables and continuity of $\mathcal{W}_u$ in $u$ (Theorem  \ref{thm:Stability under control perturbations}). The convergence of energies \eqref{eq:limit-optimal-control} results from the $\Gamma$-liminf inequality combined with the recovery sequence property of minimizers.
\end{proof}

\begin{theorem}\label{thm:Stability of optimal controls}
 Assume the second variation $\delta^2\mathcal{W}_u$ is positive definite at $u^*$. Then the optimal control is locally unique and depends smoothly on perturbations of model parameters such as $\lambda$ and the background Finsler metric $F$.
\end{theorem}

\begin{proof}
Positivity of the second variation implies strict convexity of the reduced functional $\mathcal{W}_u$ in a neighborhood of $u^*$. Let $\mathcal{P}(F,\lambda)$ denote the parameter-to-functional map. The implicit function theorem in Banach spaces \cite{Evans2010} ensures smooth dependence of $u^*$ on $(F,\lambda)$ provided that $D_u^2\mathcal{W}_u$ is invertible at $u^*$. Local uniqueness follows immediately from strict convexity.
\end{proof}

\section*{Conclusion}

This work establishes a rigorous analytical and geometric framework for Ginzburg-Landau systems within a control-translated Tonelli-Finsler setting. By embedding a smooth control field $u(x)$ as a translation in the tangent bundle, we introduced a deformed Tonelli Lagrangian that preserves Legendre duality and generates a uniformly elliptic operator $\Delta_u$. The corresponding energy functional $E_u$ remains convex, coercive, and geometrically well defined.

A detailed $\Gamma$-convergence analysis yields the renormalized energy $W_u$ governing vortex interactions under control translation, clarifying how the control field modifies both the Green kernel and the self-interaction potential. The results demonstrate that control acts intrinsically through geometric deformation of the underlying metric, rather than as an external force, thereby influencing both the microscopic vortex configuration and the macroscopic dynamics of the system.

The framework unifies Finsler geometry, convex analysis, and variational control into a coherent theoretical structure. It extends previous Finsler-Ginzburg-Landau models by providing a precise geometric mechanism for feedback, stabilization, and manipulation of vortices in anisotropic superconducting media. Future research directions include time-dependent Tonelli flows, optimal feedback control, and computational realizations in engineered anisotropic superconductors.

\section*{Data Availability}
No datasets were generated or analyzed during the current study.


\end{document}